\documentclass[12pt,english,american]{article}

\usepackage{hyperref,amsmath,amsfonts,amssymb,amsthm,mathrsfs,enumerate,bm,verbatim,textcomp}
\usepackage{graphics}

\usepackage{color}

\theoremstyle{plain}\newtheorem{theorem}{Theorem}[section]
\theoremstyle{plain}
\theoremstyle{plain}\newtheorem{proposition}[theorem]{Proposition}
\theoremstyle{plain}\newtheorem{lemma}[theorem]{Lemma}
\theoremstyle{definition}
\theoremstyle{plain}

\theoremstyle{remark}\newtheorem{remark}{Remark} 
\theoremstyle{definition}\newtheorem{def:and:lemma}[theorem]{Definition and Lemma}

\newcommand{\N}{\mathbb{N}}
\newcommand{\R}{\mathbb{R}}

\newcommand{\cF}{\mathcal{F}}
\newcommand{\cN}{\mathcal{N}}
\newcommand{\cH}{\mathcal{H}}

\newcommand{\fH}{\mathfrak{H}}

\newcommand{\eps}{\varepsilon}

\renewcommand{\phi}{\varphi}
\renewcommand{\Re}{\mathrm{Re}}

\newcommand{\aux}{\mathrm{aux}}
\newcommand{\Bog}{\mathrm{Bog}}
\newcommand{\BF}{\mathrm{BF}}

\newcommand{\eed}{\color{black}}

\newcommand{\bra}{\langle}

\newcommand{\ket}{\rangle}

\allowdisplaybreaks[1]

\title{Dynamics of a tracer particle interacting with excitations of a Bose-Einstein condensate}

\date{\today}

\author{
Jonas Lampart
\footnote{ Ludwig-Maximilians-Universit\"at, Mathematisches Institut, Theresienstr.\ 39, {80333} M\"unchen, Germany.}
\footnote{CNRS \& Laboratoire Interdisciplinaire Carnot de Bourgogne (UMR 6303), 9. Av. A. Savary, 21078 Dijon Cedex. E-mail: {\tt lampart@math.cnrs.fr}}
,
Peter Pickl\footnotemark[1]
\footnote{Eberhard-Karls-Universität, Mathematisches Institut, Auf der Morgenstelle 10, 72076 Tübingen, Germany. E-Mail: {\tt p.pickl@uni-tuebingen.de }}
}

\begin{document}

\maketitle
.
\begin{abstract}
We consider the quantum dynamics of a large number $N$ of interacting bosons coupled a tracer particle, i.e.~a particle of another kind, on a torus. We assume that in the initial state the bosons essentially form a homogeneous Bose-Einstein condensate, with some excitations. With an appropriate mean-field scaling of the interactions, we prove that the effective dynamics for $N\to \infty$ is generated by the Bogoliubov-Fröhlich Hamiltonian, which couples the tracer particle linearly to the excitation field.
\end{abstract}

\section{Introduction}

The reaction of a many-particle system to small perturbations -- for example the introduction of a different particle -- carries detailed information about its global properties.
An interesting example is a Bose-Einstein condensate with an impurity particle. Apart from the obvious motivation of understanding this state of matter, Bose-Einstein condensates in the laboratory have the advantage that many of their parameters are controllable, which allows for comparison of theory and experiment over a large range paramter values. Systems of Bose-Einstein condensates with few impurities can now be realised using cold atoms and promise numerous applications~\cite{Zipkes2010}. Theoretical approches to their description are discussed in~\cite{grusdt2016}.


In this manuscript we will perform a mathematical analysis of the dynamics of such an impurity entering a Bose gas at zero temperature. In keeping with the convention of~\cite{Froehlich2011, deckert2016}, and since one motivation is to use the impurity to probe the condensate, we will call it ``tracer'' particle in what follows.

We will show that, for a certain choice of parameters and for a Bose-Einstein condensate of constant density, the tracer particle will effectively interact with the excitations in the gas, which will be of finite number and are governed by the dynamics proposed by Bogoliubov in his seminal work~\cite{bogoliubov1947}. Thus the system is -- for a large number of bosons -- in good approximation described by the Bogoliubov-Fröhlich Hamiltonian, given in \eqref{eq:H-BF} below.
Our result is in agreement with recent findings by Mysliwy and Seiringer, who analyzed the spectral properties of the same system~\cite{Mysliwy2020}.

In order to get an interaction which is governed by the excitations in the condensate we have to assume that the condensate state has constant density -- at least in the interaction region of the tracer particle. Otherwise, the effect coming from particles in the condensate state -- which is at zero temperature the vast majority -- would be dominant. Assuming constant density, this leading contribution of the interaction will also be a constant which simply leads to a global change in phase for the dynamics, and thus does not influence the motion of the tracer.
To simplify the technicalities we assume that our system lives on a three dimensional torus of fixed length. However, we expect that this can be generalized. This generalization has to ensure, of course, that density stays constant in the range of the tracer particle for the full time scale of our observation, which amounts to considering gases of large volume. 
  
The configuration space of one particle is the $d$-dimensional torus of side-length one, $T^d=\R^d/\mathbb{Z}^d$. The wave function of a system of $N$ bosons and one tracer particle in $T^d$ is then an element of
\begin{equation}
 \cH_N=L^2(T^d) \otimes L^2(T^d)^{\otimes_s N},
\end{equation}
where $\otimes_s$ denotes the symmetric tensor product.

The dynamics of the system is described by the Schrödinger equation on $\cH_N$
\begin{equation}\label{eq:schroe}
\left\{ 
\begin{aligned}
i\frac{d}{dt}\Psi(t)&=H_N \Psi(t)\\
\Psi(0)&=\Psi_N,
\end{aligned}
\right.
\end{equation}
with Hamiltonian
\begin{equation}\label{eq:hamiltonian-gen} 
H_N=-\frac1{2m} \Delta_x-\sum_{j=1}^N\Delta_{y_j} + g_{B} \sum_{j\neq k=1}^N V(y_j-y_k)+g_{I}\sum_{j=1}^N W(x-y_j) \;,
\end{equation}
where $x$ denotes the position of the tracer and $y_1, \dots, y_n$ those of the bosons.  


As an $N$-boson initial state $\Psi_N$ one can think of an $N$-fold product of the normalized constant on the torus times an $L^2$-function of $x\in T^d$ describing the initial state of the tracer. More generally, this could be modified by adding some ``excitations'', i.e. factors orthogonal to the constant function of $y\in T^d$, as long as their number stays finite as $N\to \infty$ (see~\eqref{eq:initial} below for the precise condition).

The coupling constants $g_B$ of the boson interaction and $g_I$ of the interaction with the impurity will depend on the number $N$ of bosons. The precise form of the dependence has to be chosen appropriately to get the desired effect, i.e. a system where the interaction with excitations is of order one.

\subsection{Mean-field equations and corrections}

The effective description of many-body systems is an old and yet vivid research area of mathematical physics. Hepp~\cite{hepp} and later Ginibre and Velo~\cite{ginibre1979,ginibre1979_2} as well as Spohn~\cite{spohn1980} proved that the dynamics of a many-body Bose gas at zero temperature is, in the weak coupling limit, described by the Hartree equation in the sense that the one-particle marginal density of the $N$-body wave function converges to a pure state given by the solution of the Hartree equation as $N$ tends to infinity.
The latter can be derived by replacing the interaction of the many-body system by the respective mean-field, effectively turning the $N$-body description into a one-body dynamics. Subsequently, many improvements have been found, for example generalizations to a larger class of interactions, and different notions of convergence (see~\cite{lewin2015_3} for a review). The dilute (Gross-Pitaevskii) limit, where interactions are of order one but rare, is a variation of this theme that has also been studied extensively (see~\cite{benedikter_lec}).

More recently, there have also been many works on the fluctuations around the mean-field behavior.
Grillakis, Margetis and Machedon proved the validity of Bogoliubov theory in a weak-coupling situation with mathematical rigor~\cite{grillakis2010,grillakis2011,grillakis2013}. They showed that the number of excitations is finite in this case and that the convergence of the $N$-body wave function towards the respective solution of the Bogoliubov time evolution holds in $L^2$-norm. This implies also that the errors are small compared to the effects of Bogoliubov's theory.  
Their results have been generalized by many other authors in different situations
\cite{brennecke2017_2,grillakis2017,kuz2017,lewin2015,mitrouskas2016,nam2011,nam2015,nam2017,nam2017_2,petrat2017}.
Corrections beyond the Bogoliubov dynamics were studied in~\cite{paul2019,corr}, including explicit formulas for higher contributions~\cite{QF}.

 In  parallel to these results concerning the dynamics, there are also numerous results  related to the static question, which show closeness of the $N$-body ground state to the minimizer of an effective functional (see~\cite{lewin2015_3}).  Bogoliubov theory has also been justified in this context~\cite{seiringer2011,grech2013, derezinski2014,seiringer2014, lewin2015_2, pizzo2015,pizzo2015_2,pizzo2015_3, boccato2017}, and allows for a resolution of the  spectrum close to the ground state~\cite{derezinski2014,grech2013,seiringer2011, lewin2015_2, boccato2018}.
Recently, higher order corrections~\cite{spectrum} and the interaction with a tracer particle~\cite{Mysliwy2020} have also been considered in the static context.

\subsubsection{Heuristics of the mean-field scaling}

In this article we consider the mean-field scaling for the boson interactions,  i.e. the situation where $g_B$ is of order $N^{-1}$. With a weak coupling of this form,  we can expect the interaction energy of the $N$ bosons, which is roughly of order $N^2 g_B$, to be of the same order as their kinetic energy. 
In models for dilute gases this is complemented by a scaling of the potential $V_N$ such that $V_N$ tends to $\delta$ weakly as $N\to \infty$. We will not consider this situation here. We only remark that, in addition to the difficulties in the derivation of the Bogoliubov Hamiltonian, it will be necessary to take into account the renormalization of the limiting model (see Remark~\ref{rem:dilute}).

As explained above, we are interested in a tracer particle entering the Bose gas. 
We want to consider a scaling, respectively coupling, which is such that the interaction of the tracer is with the excitations from the ground state. Since the vast majority of the bosons will be part of the condensate, the leading contribution to the energy of the tracer will generally come from the condensate and this should be of order $g_I N_0 \sim g_IN$, where $N_0$ is the number of bosons in the condensate.  However, since the condensate density is assumed to be constant, this contribution will be trivial, i.e.~a constant.
The next order would be the back-reaction on the tracer of small changes to the condensate, i.e. excitations. Since this is a second order effect one can expect it to be of order $g_I^2 N$. We thus assume that $g_I$ is of order $N^{-1/2}$, which will also ensure that the number of excitations coming from the interaction with the tracer is of the same order as the number of excitations we have without tracer particle (i.e. of order one).
This is exactly the same scaling as in~\cite{Mysliwy2020}.

The articles~\cite{deckert2016, chen2019} considered the different situation where the interaction of the tracer with the Bose gas is such that the density of the gas is changed significantly. To achieve this, the coupling constant between tracer and gas is chosen to be one. Thus the back-reaction of the gas particles and the tracer  will be of order $N$ (respectively of the order of the density if the gas has large volume). In order to have a nice limiting regime the mass of the tracer-particle in~\cite{deckert2016, chen2019} was assumed to be of order $N$, resulting in an acceleration  of order one.

 \subsubsection{Heuristics of Bogoliubov theory}

 We now briefly explain the heuristic derivation of the Bogoliubov-Fröhlich Hamiltonian (see also~\cite[Sect.3.3]{grusdt2016}), which we will subsequently make rigorous.
 Let $\phi_0\equiv 1$ denote the constant function on the torus, which we think of as the condensate state, and extend this to a complete orthogonal system $(\phi_j)_{j\in \N_0}$ of the one-particle Hilbert space $L^2(T^d)$ (of course, one can take the $\phi_j$ to be plane waves, but for our purposes this is not necessary and we prefer to retain this generality).
  Let  $a_j^*:=a^*(\phi_j)$ and $a_j:=a(\phi_j)$ be the creation, respectively annihilation, operators for the state $\phi_j$ (see~\cite[Sect.X.7]{ReSi2} for definitions of these, and related, objects).

The boson-Hamiltonian with coupling $g_B=N^{-1}$ can then be written as

  \begin{equation}\label{eq:Hsecond}
  \begin{aligned}
  H_B=& d\Gamma(-\Delta) + N^{-1}\sum_{j, k,\ell,m=0}^\infty V_{jklm} a_j^*a_k^*a_\ell a_m
  \end{aligned}
  \end{equation}
  with
  \begin{equation}\label{eq:V proj}
  V_{jk\ell m}=\int_{T^d} \int_{T^d} \bar\phi_j(y)\bar\phi_k(y')V(y-y')\phi_\ell(y')\phi_m(y) dy'dy\;.
  \end{equation}
 
%
Bogoliubov's idea~\cite{bogoliubov1947} is that, since the vast majority of particles will be in the condensate state, the operators $a_0^*$ and $a_0$ will give contributions of order $\sqrt{N}$, whereas $a_j^*$ and $a_j$ with $j\geq1$ give contributions of order one.
Consequently, all terms with less than two of the indices $j,k,\ell,m$ equal to zero are small.
Replacing then $a_0^*$ and $a_0$ by $\sqrt{N}$ in the remaining ones and assuming that $V$ is even and $\int_{T^d} V=0$ (which corresponds to a constant shift in energy and eliminates all terms where more than two of the indices  $j$, $k$, $\ell$ and $m$ are equal to zero -- see also Lemma~\ref{lem:V trafo} below), one arrives at the Bogoliubov Hamiltonian
\begin{equation}\label{eq:HBog}
\begin{aligned}
H^\Bog=& d\Gamma(-\Delta) + 2\sum_{j, k=1}^\infty V_{j0k0} a_j^*a_k
+ \sum_{j, k=1}^\infty  V_{jk00}  a_j^*a_k^*
+ \sum_{j, k=1}^\infty  V_{00jk}  a_ja_k\;.
\end{aligned}
\end{equation}

If we apply the same reasoning to the interaction of the bosons with the tracer,
we write 
\begin{equation}\label{eq:W decomp}
\sum_{n=1}^N W(x-y_n) =  \sum_{j,k=0}^\infty a_j^*a_k (W_x)_{j,k} ,
\end{equation}
with $W_x(y)=W(x-y)$, $(W_x)_{j,k}= \int\bar\phi_j(y)\phi_k(y)W_x(y)d y$.
Using that $(W_x)_{0,0}(x) = \int W=0$, and then
setting $g_I=N^{-1/2}$, dropping the terms with $j\neq 0\neq k$ and replacing $a_0^*$ and $a_0$ by $\sqrt{N}$ as above, we arrive at (see also Lemma~\ref{lem:W trafo} below)
\begin{equation}
 \sum_{j=1}^\infty \Big( a_j (W_x)_{0,j} + a_j^* (W_x)_{j,0} \Big)=  a(W_x) +a^*(W_x).
\end{equation}

Adding this interaction as well as the kinetic energy of the tracer to the Bogoliubov Hamiltonian, we obtain the Bogoliubov-Fröhlich Hamiltonian
\begin{equation}\label{eq:H-BF}
H^\BF = -\frac{1}{2m} \Delta_x  + H^\Bog  + a^*(W_x) + a(W_x).
\end{equation}

\begin{remark}
 An important feature of the original Bogoliubov Hamiltonian is that it is a quadratic expression in the operators $a,a^*$. 
There is a unitary Bogoliubov transformation  $U_\Bog$ (that amounts to changing the creation/annihilation operators) such that

\begin{equation}
U_\Bog H^\Bog U_\Bog^* = d \Gamma (E) + E_0
\end{equation}
for a one-particle operator $E$ and a constant $E_0$.
The excitations  can thus be described by a non-interacting theory and, in the case of the torus, the energy levels can be computed explicitly (see e.g.~\cite{Mysliwy2020}).

We emphasize that the Bogoliubov-Fröhlich Hamiltonian does not share this feature. Even though the expression~\eqref{eq:H-BF} seems to be quadratic, the interaction depends on the position $x$ of the tracer particle, which also appears in the Laplacian. Performing an $x$-dependent Bogoliubov transformation will thus not yield a simple result (alternatively, introducing a field for the impurities and viewing $H^\BF$ as the restriction to the one-impurity space, one sees that the interaction is in fact cubic in the creation and annihilation operators, cf.~\cite[Sect.3.2]{grusdt2016}).
One can, of course, transform $H^\BF$ using $U_\Bog$, which yields an operator of the form
\begin{equation}\label{eq:HBF trafo}
 (1\otimes U_\Bog) H^\BF (1\otimes U_\Bog^*) = -\frac{1}{2m} \Delta_x + d\Gamma(E) + a^*(\tilde W_x) + a(\tilde W_x) + E_0,
\end{equation}
with a suitably transformed interaction $\tilde W$ (see~\cite[Sect.3.3]{grusdt2016} and~\cite[Eq.(1.10)]{Mysliwy2020} for explicit expressions).
%
Operators of the form~\eqref{eq:HBF trafo}, with general $E$ and $W$, which feature the interaction of a particle (or several) with a quantum field by a linear coupling have been studied in many variants and are sometimes referred to as Fröhlich or polaron Hamiltonians.
They include the original Fröhlich model~\cite{Hfrohlich1937} and the Nelson model~\cite{nelson1964}.
One of the main motivations of the present work is to show rigorously that such a model arises as an effective description of a many-particle system.
\end{remark}

\begin{remark}\label{rem:dilute}
 If, instead of the mean-field regime, we were to consider a dilute $N$-particle system (as in~\cite{boccato2017}) then, as an intermediate step, we would find a Hamiltonian of the same form as $H^\BF$, but with $W=W_N$ depending on $N$ and converging to $\delta$ weakly as $N\to \infty$. Since $a^*(\delta)$ is not a densely defined operator, it is not immediately clear how to define the limiting dynamics. It was shown in~\cite{lampart2020} that a renormalized version of $H^\BF$ with $W=\delta$ can be constructed in this case and that its unitary group can be approximated using operators with ultraviolet cut-off, up to a divergent phase. In a dilute system the particle number $N$ would function as an effective ultraviolet cutoff, so we expect that in this case the phase would have to be modified accordingly in order to obtain the effective dynamics generated by the renormalized Bogoliubov-Fröhlich Hamiltonian.
\end{remark}

\subsection{The condensate-excitation representation}

In order to make the heuristics of the previous section rigorous, we will represent functions in $\cH_N$ by decomposing them into their components along the (constant) condensate wave-function $\phi_0$ and in the orthogonal complement. The rigoros implementation of this idea was pioneered by Lewin, Nam, Serfaty and Solovej~\cite{lewin2015_2} and we closely follow their presentation. It naturally gives rise to Fock space and thus the representation of the excitations by a quantum field.

Recall that $\phi_0\equiv 1$ denotes the constant function on $T^d$ and $(\phi_j)_{j\in \N_0}$ a complete orthonormal system in $L^2(T^d)$ with $\phi_j\in H^2(T^d)$, $j\in \N$.
Let 
\begin{equation*}
 P\psi= \int_{T^d}\psi(y) dy
\end{equation*}
be the projection to $\mathrm{span}\{\phi_0\}=:\fH_0$ in $L^2(T^d)$ and $Q=1-P$.
The function $\phi_0$ plays the role of the condensate wave-function, while $\fH_+:=QL^2(T^d)$ is the Hilbert space of an excitation.
The functions  $(\phi_j)_{j\in \N}$, with $j>0$, thus form a complete orthonormal system in $\fH_+$.
Let $\Psi \in L^2(T^d)^{\otimes_s N}$, then we can write
\begin{align}
 \Psi= (P+Q)^{\otimes N} \Psi = \sum_{j=0}^{N} \phi_0^{\otimes (N-j)} \otimes_s \Phi^{(j)}
\end{align}
with $\Phi^{(j)}\in \fH_+^{\otimes_s j}$. 
We can thus represent $\Psi\in L^2(T^d)^{\otimes_s N}$ uniquely by the sequence
\begin{equation}
 \big(\Phi^{(n)}\big)_{n\leq N } \in \bigoplus_{n=0}^N \fH_+^{\otimes_s n} =:\cF_+^{\leq N}.
\end{equation}
We denote by $\cF_+=\Gamma(\fH_+)$ the Fock space of excitations, and by $\cF_+^{\leq N}$ the truncated Fock space, defined as above.

The $N$-boson space $\cH_N$ also contains an additional factor for the tracer particle. We define
\begin{equation}
 \cH_+=:L^2(T^d)\otimes \cF_+, \qquad \cH_+^{\leq N}=:L^2(T^d)\otimes \cF_+^{\leq N}.
\end{equation}
Leaving the first tensor factor untouched, we obtain a unitary
\begin{equation}
 U_N : \cH_N \to \cH_+^{\leq N} \subset \cH_+.
\end{equation}

By viewing $L^2(T^d)^{\otimes_s N}$ as embedded in the Fock space $\cF=\Gamma(L^2(T^d))$, the unitary can be expressed by the formula (see~\cite[Eq.(4.6)]{lewin2015_2})
\begin{equation}\label{eq:U_N form}
 U_N \Psi=\bigoplus_{j=0}^N Q^{\otimes j} \bigg(\frac{a_0^{N-j}}{\sqrt{(N-j)!}}\Psi\bigg).
\end{equation}

Denote by $\cN_+:=d\Gamma(Q)$ the number operator in $\cF_+$. Then from~\eqref{eq:U_N form} one easily deduces the following identities  for $f,g\in \fH_+$ (see~\cite[Prop.14]{lewin2015_2}):
\begin{equation}\label{eq:a trafo}
\begin{aligned}
 U_Na_0^*a_0 U_N^* &= N-\cN_+\\
 U_N a^*(f) a_0 U_N^*&=a^*(f)\sqrt{N-\cN_+}\\
 U_N a_0^*a(f) U_N^*&=\sqrt{N-\cN_+} a(f) \\
 U_Na(f)^*a(g) U_N^*&= a(f)^*a(g)
\end{aligned}
\end{equation}
Note that on the left hand side we only consider number-preserving combinations of creation and annihilation operators, so these expressions can be viewed as operators on $\cH_N$.

We now apply this transformation to the interaction terms in the Hamiltonian.

\begin{lemma}\label{lem:V trafo}
Let $(\phi_n)_{n\in \N_0}$ be as above and assume that $V\in L^2(T^d)$  is an even function with $\int_{T^d} V(y)dy=0$ and relatively $\Delta$-bounded as an operator.
The following identity holds in the sense of closed operators
\begin{align*}
 U_N &\bigg(\sum_{j,k=1}^N V(y_j-y_k)\bigg) U_N^*
 = 2\sum_{j, k \in \N} V_{j0k0} a_j^*(N-\cN_+)a_k  
\\
&+\sum_{j, \in \N} V_{jk00} a_j^*\sqrt{N-\cN_+} a_k^* \sqrt{N-\cN_+} 
+\sum_{j, \in \N} V_{00jk} \sqrt{N-\cN_+} a_j  \sqrt{N-\cN_+}a_k
\\
 &+2\sum_{j,k, \ell\in \N} \Big(V_{jk\ell0} a^*_j \sqrt{N-\cN_+} a^*_k  a_\ell + V_{0jk\ell} a^*_j  a_k \sqrt{N-\cN_+} a_\ell\Big)
 +\sum_{j,k, \ell,m \in \N} V_{jk\ell m} a^*_j a^*_k a_\ell a_m
\end{align*}
with $V_{jk\ell m}$ defined by~\eqref{eq:V proj}.
\end{lemma}

\begin{proof}
 We write the mulitplication operator $\sum_{j,k} V(y_j-y_k)$ on $L^2(T^d)$ as in~\eqref{eq:Hsecond} (where $V_{jk\ell m}$ is well-defined since the $\phi_j\in H^2$ and $V$ is $\Delta$-bounded). We then apply the transformation $U_N$ using the identities~\eqref{eq:a trafo}. Note that, since $V$ is even, $V_{0j0k}=V_{j0k0}$ are real. In order to transform the quartic expressions in $a_j$, $a_j^*$, these must be grouped into number-preserving pairs, e.g. for $j,k, \ell>0$, 
 \begin{equation}
  U_N a_j^*a_k^*  a_\ell a_0 U_N^*= U_N a_j^* a_0 U_N^* U_N a_k^*  a_\ell  U_N^*= a_j^* \sqrt{N-\cN_+} a_k^* a_\ell.
 \end{equation}
The claim then follows by noting that $V_{jk\ell 0}=V_{kj0\ell}$ and that $\int V=0$ implies $V_{j00k}=0=V_{0jk0}$ for $j, k\in \N_0$.
\end{proof}

\begin{lemma}\label{lem:W trafo}
Assume that $W\in L^2(T^d)$ be $\Delta$-bounded with $\int_{T^d} W(y)dy=0$ and set $W_x(y):=W(x-y)$.
The following identity holdsin the sense of closed operators
\begin{align*}
 U_N &\bigg(\sum_{j=1}^N W(x-y_j)\bigg) U_N^*
 = \sqrt{N-\cN_+} a\big(W_x\big) + 
 a^*\big(W_x\big) \sqrt{N-\cN_+}
 + d \Gamma (QW_x Q).
\end{align*}
\end{lemma}
\begin{proof}
We can rewrite $W_x$ as in~\eqref{eq:W decomp}.
The claim then follows from the identities~\eqref{eq:a trafo} by noting that
\begin{equation}\label{eq:W recompose}
\begin{aligned}
(W_x)_{00}&=PW_xP=\int_{T^d} W(x-y) dy =0, \\
  \sum_{k=1}^\infty  (W_x)_{k,0}a_k^*&=\sum_{k=1}^\infty a^*(\phi_k) \langle \phi_k, W_x \phi_0\rangle = a^*\big(QW_x\phi_0\big)=a^*(W_x),\\
  \sum_{k, \ell=1}^\infty a_k^* a_\ell (W_x)_{k,\ell} &= d\Gamma(QW_xQ).
\end{aligned}
\end{equation}

\end{proof}

\section{Main Result}

We assume that the $N$-particle initial state $\Psi_N\in \cH_N$ is a member of a sequence such that
\begin{equation}\label{eq:initial}
 \lim_{N\to \infty} U_N \Psi_N=:\Phi
\end{equation}
exists in $\cH_+$. 
This assumption means that the number of excitations out of the condensate in the initial states $\Psi_N$ remains finite \textit{with probability one} as $N\to \infty$. Note, however, that the \textit{expected number} of excitations $\langle \cN_+ U_N\Psi_N, U_N\Psi_N \rangle$ may diverge.

Recall that the expression of $H^\BF$ in~\eqref{eq:H-BF} is given by $H_0:=-\frac{1}{2m}\Delta_x + d\Gamma(-\Delta)$ plus several terms that are linear or quadratic in creation and annihilation operators. If $V, W\in L^2(T^d)$, the latter are bounded relative to $\cN_+$ (see Proposition~\ref{prop:dom BF}). Consequently, $H^\BF$ is a well defined operator on $D(H_0)$, since  on $\cF_+$
 \begin{equation}
  d\Gamma(-\Delta) \geq 4\pi^2 \cN_+.
 \end{equation}
In fact, $H^\BF$ is essentially self-adjoint on $D(H_0)$, see Proposition~\ref{prop:dom BF}. 
Thus the time evolution under $H^\BF$, $e^{-i H^\BF t} \Phi$, is well-defined and we may compare it to the time evolution with $H_N$ using $U_N$.

If we make the assumption that $\Phi\in D(H_0)$, then we obtain a quantitative estimate proving the closeness of the two evolutions for large $N$.

\begin{theorem}\label{thm:quant}
Assume that $V, W\in L^2(T^d, \R)$ are infinitesimally bounded relative to $-\Delta$ and satisfy $V(-y)=V(y)$ and $\int_{T^d} V=0=\int_{T^d} W$.
There exists $v>0$ such that for all $\Phi \in  D(H_0)$  there is $K>0$ such that for all $N\in \N$ and $t\in \R$
	\begin{equation*}
	\big \|U_N e^{-i H_N t} U_N^*\Phi^{\leq N} - e^{-i H^\BF t} \Phi\big\|_{\cH_+}\leq K e^{v|t|} N^{-1/4}\;.
	\end{equation*}
\end{theorem}
The proof of Theorem~\ref{thm:quant} will be the content of the next section.
Note that in this statement only $K$ depends on the intial states. By density of $D(H_0)$ in $\cH_+$, the statement can be extended to any $\Phi\in \cH_+$, but without an explicit rate of convergence.

\begin{theorem}\label{thm:gen}
 Assume that $V, W\in L^2(T^d, \R)$ are infinitesimally bounded relative to $-\Delta$ and satisfy $V(-y)=V(y)$ and $\int_{T^d} V=0=\int_{T^d} W$. Then for every $\Phi \in \cH_+$ and  and every sequence $\Psi_N$, $N\in \N$, with $\lim_{N\to \infty} U_N\Psi_N=\Phi$ we have
\begin{equation*}
\lim_{N\to \infty}\big \|U_N e^{-i H_N t} \Psi_N - e^{-i H^\BF t} \Phi\big\|_{\cH_+}=0 
\end{equation*}
uniformly in $t$ on compact subsets of $\R$.
\end{theorem}
\begin{proof}
This follows from  Theorem~\ref{thm:quant} above by an approximation argument. Let $\Phi\in \cH_+$ be given and $T>0$, $\eps>0$. Let $\widetilde \Phi\in D(H_0)$ with $\|\Phi-\widetilde\Phi\|_{\cH_+} < \eps$ and set $\widetilde \Psi_N=U_N^*\widetilde \Phi^{\leq N}$. Since $\lim_{N\to \infty} U_N\Psi_N=\Phi$, we have for  sufficiently large $N$
 \begin{equation}
  \|\Psi_N-\widetilde \Psi_N\|_{\cH_N} \leq \|U_N\Psi_N- \Phi\|_{\cH_+} + \|\Phi-\widetilde\Phi\| + \|\widetilde \Phi^{>N}\|< 3\eps.
 \end{equation}
By unitarity of the time evolutions we thus have 
\begin{equation}
 \big \|U_N e^{-i H_N t} \Psi_N - e^{-i H^\BF t} \Phi\big\|_{\cH_+} < 4\eps + \big \|U_N e^{-i H_N t} \widetilde\Psi_N - e^{-i H^\BF t} \widetilde\Phi\big\|_{\cH_+}.
\end{equation}
By Theorem~\ref{thm:quant} the last term is bounded for $|t|\leq T$ by
\begin{equation}
 \big \|U_N e^{-i H_N t} \widetilde\Psi_N - e^{-i H^\BF t} \widetilde\Phi\big\|_{\cH_+} \leq K e^{v T}N^{-1/4},
\end{equation}
so it is smaller than $\eps$ for sufficiently large $N$.
\end{proof}

\section{Proof of convergence}

Throughout this section we assume the hypothesis of Theorem~\ref{thm:quant}.
As a fist step we will prove that the time evolution $U_N e^{-iH_N t} U_N^*$ is well approximated by the time evolution generated by a truncation of $U_N H_N U_N^*$, where terms that are  more than quadratic in the creation and annihilation operators are neglected. We collect the leading-order terms in an auxiliary operator $H^\aux$.
This operator acts on $\cH_+^{\leq N}$ and is naturally extended to $\cH_+$ by zero. Since $\langle \phi_k, \Delta \phi_0\rangle=0= \langle \phi_0, \Delta \phi_k\rangle$ for $k\in \N_0$, we have
\begin{equation}
-U_N \sum_{j=1}^N \Delta_{y_j} U_N^* =  d\Gamma^{\leq N}(-\Delta),
\end{equation}
where $d\Gamma^{\leq N}(-\Delta)$ denotes the projection of $d\Gamma(-\Delta)$ to $\cH_+^{\leq N}$. In view of Lemmas~\ref{lem:W trafo},~\ref{lem:V trafo}, we  define $H_N^\aux$ by the expression
\begin{align}
H^\aux_N:= &- \frac1{2m} \Delta_x + d\Gamma^{\leq N}(-\Delta) +  \frac{2}{N}\sum_{j, k>1} V_{j0k0} a_j^*(N-\cN_+)a_k \notag \\
&+ \frac{1}{N}\sum_{j, k>1} V_{jk00} a_j^*\sqrt{N-\cN_+} a_k^* \sqrt{N-\cN_+} \notag\\
&+\frac{1}{N}\sum_{j, k>1} V_{00jk}  \sqrt{N-\cN_+}
a_k\sqrt{N-\cN_+} a_k\notag \\
& + \frac{1}{\sqrt{N}} a^*(W_x)\sqrt{N-\cN_+} +  \frac{1}{\sqrt{N}} \sqrt{N-\cN_+}a(W_x)\label{eq:haux}
\end{align}
on $D(H_N^\aux)=U_ND(H_N)\subset \cH_+^{\leq N}$. Note that $H_N^\aux$ is self-adjoint by Lemma~\ref{lem:dom aux} and that $H^\BF$ is obtained from $H_N^\aux$ by simply replacing $\sqrt{1-\cN_+/N}$ by one and extending to $\cH_+$.

To prove closeness of $e^{-i H^\aux_N t}$ and $U_N e^{-i H_N t}U_N^*$, and later $e^{-i H^\BF t}$ and $e^{-i H^\aux_N t}$, the crucial ingredient is control of the number of excitations for all times. 

\begin{lemma}\label{lem:alpha}
For any $\Phi\in D(H_0)$ set
 \begin{align*}
  \alpha(t):=&\| (\cN_++1) U_N e^{-i H_N t}U_N^* \Phi^{\leq N}\|^2 \\
  \alpha^\aux(t):=&\| (\cN_++1)  e^{-i H^\aux_N t} \Phi^{\leq N}\|^2\\
  \alpha^\BF(t):=&\| (\cN_++1)  e^{-i H^\BF t} \Phi\|^2.
 \end{align*}
There exists a constant $v$ such that for all $\Phi\in D(H_0)$, $N\in \N$ and $t\in \R$ we have
\begin{equation*}
 \alpha^\bullet(t) \leq  \alpha^\bullet(0) e^{v|t|}
\end{equation*}
for $\bullet\in \{\emptyset, \aux, \BF\}$.
\end{lemma}

\begin{proof}
Since $\Phi\in D(\cN_+)$, the statement clearly holds for $t=0$. We will use Grönwall's Lemma to obtain a bound for all $t\geq 0$ (the proof for $t\leq 0$ is the same).

 We first prove the claim for $\alpha$ and $\alpha^\aux$ (since $H_N^\aux$ is obtained from $U_N H_NU_N^*$ by dropping some terms the proof for $\alpha^\aux$ is contained in the one fo $\alpha$).
  By Lemma~\ref{lem:dom aux}, we have $\Phi^{\leq N}\in U_ND(H_N)= D(H^\aux)$. Thus
 \begin{align}
 \frac{d \alpha}{dt}(t)&=\frac{d }{dt}\Big\langle U_N e^{-i H_N t}U_N^*\Phi^{\leq N}, (\cN_++1)^2 U_N e^{-i H_N t}U_N^*\Phi^{\leq N}\Big\rangle \notag \\
 &= \Big\langle U_N e^{-i H_N t}U_N^*\Phi^{\leq N},i[U_N H_N U_N^*, (\cN_++1)^2] U_N e^{-i H_N t}U_N^*\Phi^{\leq N}\Big\rangle .
 \end{align}
With $[U_N H_N U_N^*, (\cN_++1)^2]=(\cN_++1) [U_N H_N U_N^*, \cN_+] + [U_N H_N U_N^*, \cN_+](\cN_++1)$ and Cauchy-Schwarz we obtain
\begin{equation}
\Big| \frac{d \alpha}{dt}(t)\Big| \leq 2 \sqrt{\alpha(t)} \left\|[U_N H_N U_N^*, \cN_+]e^{-i H_N t}U_N^*\Phi^{\leq N}\right\|.
\end{equation}
We can thus prove the claim by bounding the latter norm in terms of $\sqrt{\alpha}$.

To do this, we start with the terms from $H^\aux_N$. Using that
 \begin{align}
 [a_j^* a_j,a_k^* a_\ell^*]&= (\delta_{jk} + \delta_{j\ell} ) a^*_j a^*_j \\
 [a_j^* a_j,a_k a_\ell]&= -(\delta_{jk} + \delta_{j\ell} ) a_j a_j
\end{align}
and $\cN_+=\sum_{j=1}^\infty a_j^*a_j$ we find 
\begin{align}
\Big[\cN_+, \sum_{j,k=1}^\infty V_{jk00} a_j^*a_k^* \Big]&=2 \sum_{j,k=1}^\infty V_{jk00} a_j^*a_k^* \label{eq:quad comm+}\\
\Big[\cN_+, \sum_{j,k=1}^\infty V_{00jk} a_ja_k \Big]&=-2 \sum_{j,k=1}^\infty V_{00jk} a_ja_k.\label{eq:quad comm-}
\end{align}
We thus have
\begin{align}
 [\cN_+, H_N^\aux]=& \frac{2}{N}\sum_{j, k=1}^\infty V_{jk00} a_j^*\sqrt{N-\cN_+} a_k^* \sqrt{N-\cN_+} \notag\\
&-\frac{2}{N}\sum_{j, k=1}^\infty V_{00jk}  \sqrt{N-\cN_+}
a_k\sqrt{N-\cN_+} a_k\notag \\
& + \frac{1}{\sqrt{N}} a^*(W_x)\sqrt{N-\cN_+} - \frac{1}{\sqrt{N}} \sqrt{N-\cN_+}a(W_x).
\end{align}
As $\|N^{-1/2}\sqrt{N-\cN_+}\|_{\cH_+^{\leq N}\to \cH_+^{\leq N}}\leq 1$, this operator is bounded relative to $\cN_+$ by Lemma~\ref{lem:quadratic} and~\eqref{eq:V ell2}, which gives the desired bound.
The expression for $U_N H_N U_N^*$ additionally contains the term $N^{-1/2} d\Gamma(QW_xQ)$ from the interaction (see Lemma~\ref{lem:W trafo}), which commutes with $\cN_+$. Additionally, there are the terms from the boson interaction (see Lemma~\ref{lem:V trafo})
\begin{align}
 \frac{2}{N}&\sum_{j,k, \ell=1}^\infty \Big(V_{jk\ell0} a^*_j \sqrt{N-\cN_+}a^*_k  a_\ell 
 + V_{0jk\ell} a^*_j  a_k \sqrt{N-\cN_+}a_\ell \Big),\\
 \frac{1}{N}&\sum_{j,k, \ell,m=1}^\infty V_{jk\ell m} a^*_j a^*_k a_\ell a_m.
\end{align}
The latter also commutes with $\cN_+$, so it remains to bound the commutator of $\cN_+$ with the first line.
We have from the canonical commutation relations (for the first term -- the second one yields minus the adjoint)
\begin{align}\label{eq:cubic comm}
 \Big[\cN_+, \frac{2}{N}V_{jk\ell0} a^*_j \sqrt{N-\cN_+}a^*_k  a_\ell   \Big] 
 &=-\frac{2}{N}\sum_{j,k, \ell=1}^\infty V_{jk\ell0} a^*_j \sqrt{N-\cN_+}a^*_k  a_\ell.
 \end{align}
In order to bound this operator, let $\Psi, \Xi\in \cH_+^{\leq N}$, and rewrite
\begin{align}
 &  \sum_{j,k,\ell=1}^\infty
 \big\langle \Psi, V_{jk\ell0} a_j^* \sqrt{N-\cN_+} a_k^*  a_\ell \Xi \big\rangle 
 \notag \\
 &= \int_{\left(T^d\right)^2}\sum_{j,k,\ell=1}^\infty \bar\phi_j (y) \phi_k(y^\prime) \big\langle \sqrt{N-\cN_+} a_j a_k\Psi, V(y-y^\prime)\phi_\ell (y) a_\ell \Xi \big\rangle  dy dy^\prime. 
 %
\end{align}
By the Cauchy-Schwarz inequality and the fact that $\|N^{-1/2}\sqrt{N-\cN_+}\|_{\cH_+^{\leq N}\to \cH_+^{\leq N}}\leq 1$ we thus have
 \begin{align}\label{eq:V cubic bound}
  \frac{1}{N} \Big| \sum_{j,k,\ell=1}^\infty
 \big\langle \Psi, V_{jk\ell0} a_j^* \sqrt{N-\cN_+} a_k^*  a_\ell \Xi \big\rangle \Big| 
\leq & \bigg(\int_{\left(T^d\right)^2}  \Big\| \sum_{j,k=1}^\infty \phi_j (y) \phi_k (y^\prime)a_ja_k  \Psi\Big\|^2 dydy^\prime\bigg)^{1/2} \notag  \\ 
&\left(\int_{\left(T^d\right)^2} \Big\| N^{-1} \sum_{\ell=1}^\infty   V(y-y^\prime)\phi_\ell (y) a_\ell \Xi\Big\|^2 dydy^\prime \right)^{1/2}.
 \end{align}
Now
\begin{align}\label{erste}
&N^{-1}\int_{\left(T^d\right)^2}\Big\| \sum_{\ell=1}^\infty  V(y-y^\prime)\phi_\ell (y) a_\ell \Xi\Big\|^2 dydy^\prime\nonumber\\
&\leq \sup_{y\in T^d}\left\{\int_{ T^d }V^2(y-y^\prime)dy^\prime\right\} N^{-1} \int_{T^d} \Big\|\sum_{\ell=1}^\infty  \phi_\ell (y) a_\ell \Xi\Big\|^2 dy \nonumber\\
%
%
&= \sup_{y\in T^d}\left\{\int_{ T^d }V^2(y-y^\prime)dy^\prime\right\} N^{-1} \int_{ T^d}
\sum_{\ell,m=1}^\infty\Big\langle\phi_\ell (y) a_\ell \Xi,\phi_m (y) a_m \Xi\Big\rangle dy\nonumber\\
&\leq  N^{-1} \|V\|_{L^2(T^d) }^2 \sum_{\ell=1}^\infty\Big\langle a_\ell \Xi, a_\ell\Xi\Big\rangle \leq \|V\|_{L^2(T^d) }^2 \|\Xi\|^2
\end{align}
where we used that the $\left(\phi_n\right)_{n\in\mathbb{N}}$ form an ONB and that $\sum_{k\in \N} a_k^*a_k=\cN_+$.
Similarly,
\begin{align}
&\int_{\left(T^d\right)^2}  \Big\| \sum_{j,k=1}^\infty \phi_j (y) \phi_k (y^\prime)a_ja_k  \Psi\Big\|^2 dydy^\prime\nonumber\notag\\
&=\int_{\left(T^d\right)^2}  \Big\langle \sum_{j,k,\ell,m=1}^\infty \phi_j (y) \phi_k (y^\prime)a_ja_k  \Psi,\phi_\ell (y) \phi_m (y^\prime)a_\ell a_m  \Psi\Big\rangle dydy^\prime\notag\\
&=  \Big\langle \sum_{j,k=1}^\infty a_ja_k  \Psi,a_ja_k  \Psi\Big\rangle, \label{eq:phi-aa}
\end{align}
and we have
\begin{align}\label{eq:aaaa}
	\sum_{j,k=1}^\infty\langle a_j a_ k\Psi, a_j a_ k\Psi \rangle \notag
	%
	&=  \sum_{j,k=1}^\infty\langle \Psi, (a_j^* a_j a_k^* a_ k + a_j^* [a_k^*, a_j] a_k)\Psi \rangle \notag\\
	&= \langle\Psi, \cN_+(\cN_+-1)\Psi\rangle \leq  \|\cN_+\Psi\|^2.
\end{align}
%
This implies that the commutator~\eqref{eq:cubic comm} is $\cN_+$-bounded uniformly in $N$, so, by our earlier reasoning, there is a constant $C$ such that
\begin{equation*}
 \frac{d \alpha}{dt}(t) \leq \Big|\frac{d \alpha}{dt}(t)\Big| \leq C \alpha(t). 
\end{equation*}
Thus by Grönwall's Lemma $\alpha(t) \leq \alpha(0)e^{Ct}$, which proves the claim for $\alpha$, $\alpha^\aux$.

The proof for $\alpha^\BF$ follows from the same argument, since $[\cN_+, H^\BF]$ is $\cN_+$-bounded by the identities~\eqref{eq:quad comm+},~\eqref{eq:quad comm-} and Lemma~\ref{lem:quadratic}.
Taking the growth rate $v$ to be the maximum of the constants for $\alpha$, $\alpha^{\BF}$ proves the claim.
\end{proof}

\subsection{Proof of Theorem \ref{thm:quant}}

\begin{lemma}
	\label{lem:duhamel}
	Let $\cH$ be a Hilbert space and $D\subset \cH$ a dense subspace. If $(H_1, D)$,  $(H_2, D)$ are self-adjoint operators, then for $\Psi_j\in D$, $j=1,2$, and $t\geq 0$
	\begin{equation*}
	\left\|	 e^{-iH_1t}\Psi_1-e^{-iH_2 t}\Psi_2\right\|_\cH^2 \leq   \|\Psi_1-\Psi_2\|_\cH^2 + 2\int_{0}^t\Big|  \bra e^{-iH_1s}\Psi_1,(H_1-H_2)e^{-iH_2 2}\Psi_2\ket\Big|ds.
	\end{equation*}
\end{lemma}
\begin{proof}
We have
\begin{align}
\frac{d}{dt}\left\|	 e^{-iH_1t}\Psi_1-e^{-iH_2 t}\Psi_2\right\|_\cH^2&=-2 \Re\langle e^{-iH_1t}\Psi_1,i(H_1-H_2) e^{-iH_2t}\Psi_2 \rangle, 
\end{align}
so the claim follows from the fundamental theorem of calculus.
\end{proof}

\begin{lemma}\label{lem:aux}
Assume the hypothesis of Theorem~\ref{thm:quant} and let $\Phi\in D(H_0)$, $\Psi_N=U_N^* \Phi^{\leq N}$. 
  There exists $K>0$ such that for all $N\in \N$ and $t\in \R$
	\begin{equation*}
	\big \| e^{-i H_N t} \Psi_N - U_N^* e^{-i H^\aux_N t} U_N \Psi_N \big\|\leq K N^{-1/4} e^{v |t|},
	\end{equation*}
	where $v$ is the constant of Lemma~\ref{lem:alpha}.
\end{lemma}
\begin{proof}
 By Lemma~\ref{lem:dom aux}, we have $D(H_0)\subset D(H^\aux)$ and thus $\Psi_N\in D(H_N)= U_N^*D(H^\aux)$. Using Lemma \ref{lem:duhamel} it then follows that
\begin{align}
\big \| e^{-i H_N t}& \Psi_N - U_N^* e^{-i H^\aux_N t} U_N \Psi_N \big\|^2 \notag\\
\leq & 2 \int_{0}^t\big|  \bra e^{-i H_N t} \Psi_N,(H_N - U_N^* H^\aux U_N) U_N^* e^{-i H^\aux_N t} U_N \Psi_N\ket\big|ds \notag \\
= & 2 \int_{0}^t\big|  \bra U_N e^{-i H_N t} \Psi_N,(U_N H_N U_N^* - H^\aux) e^{-i H^\aux_N t} \Phi^{\leq N}\ket\big|ds.
\label{eq:diffaux}
\end{align}
In view of Lemmas~\ref{lem:W trafo},~\ref{lem:V trafo}, we have 
\begin{align}
 U_N H_N U_N^* - H^\aux= &  \frac{1}{\sqrt{N}} d\Gamma(QW_xQ)\label{eq:H diffW}  \\
  &+ \frac{2}{N}\sum_{j,k, \ell=1}^\infty \Big(V_{jk\ell0} a^*_j \sqrt{N-\cN_+}a^*_k  a_\ell + V_{0jk\ell} a^*_j  a_k \sqrt{N-\cN_+}a_\ell \Big) \label{eq:H diff cubic}\\
&+ \frac{1}{N}\sum_{j,k, \ell,m=1}^\infty V_{jk\ell m} a^*_j a^*_k a_\ell a_m. \label{eq:H diff quartic}
\end{align}

 The first term in~\eqref{eq:H diff cubic} satisfies, by the reasoning of~\eqref{eq:V cubic bound}, \eqref{erste} and \eqref{eq:phi-aa}
\begin{align}
 \frac{2}{N}&\Big|\sum_{j,k, \ell=1}^\infty \Big\langle U_N e^{-i H_N t}\Psi_N, V_{jk\ell0} a^*_j \sqrt{N-\cN_+} a_k^* a_\ell e^{-i H^\aux_N t} \Phi^{\leq N} \Big\rangle\Big| \notag\\
 &\leq \frac{2}{\sqrt  N}\|V\|_{L^2}\|\sqrt{\cN_+}U_N e^{-i H_N t}\Psi_N\|\|\cN_+ e^{-i H^\aux_N t} \Phi^{\leq N}\|\notag\\
 &\leq \frac{2}{\sqrt  N} \|V\|_{L^2}\big(\alpha(t)\alpha^\aux(t) \big)^{1/2}.
\end{align}
The adjoint term from~\eqref{eq:H diff cubic} satisfies the same bound.

The quartic term~\eqref{eq:H diff quartic} will require some regularity of $\Phi$ (unless $V\in L^\infty$). First, we may expand
\begin{align}
  & \sum_{j,k, \ell,m=1}^\infty \Big\langle U_N e^{-i H_N t}\Psi_N, V_{jk\ell m} a^*_j a^*_k a_\ell a_m e^{-i H^\aux_N t} \Phi^{\leq N} \Big\rangle \notag \\
  &=\int_{\left(T^d\right)^2} \sum_{j,k, \ell,m=1}^\infty \Big\langle \phi_j(y)\phi_k(y^\prime)a_j a_kU_N e^{-i H_N t}\Psi_N, V(y-y^\prime)  \phi_\ell(y)\phi_m(y^\prime)a_\ell  a_m e^{-i H^\aux_N t} \Phi^{\leq N} \Big\rangle \notag dy dy^\prime.\notag \\
\end{align}
Then, noting that the Laplacian is self-adjoint and invertible on $\fH_+$, multiplying with the identity operator in the form $\Delta_y^{-1}\Delta_y$ and using the Cauchy-Schwarz inequality gives
\begin{align}
  \frac{1}{N}& \Big|\sum_{j,k, \ell,m=1}^\infty \Big\langle U_N e^{-i H_N t}\Psi_N, V_{jk\ell m} a^*_j a^*_k a_\ell a_m e^{-i H^\aux_N t} \Phi^{\leq N} \Big\rangle \Big|\notag \\
  & \leq \left(\int_{\left(T^d\right)^2}  \frac{1}{N}\Big\|\Delta_y^{-1}V(y-y^\prime)\sum_{j,k=1}^\infty\phi_j(y)\phi_k(y^\prime)a_j a_k U_N e^{-i H_N t}\Psi_N\Big\|^2 dy dy^\prime\right)^{1/2} \label{eq:quartic-t1}  \\
  &\qquad \times
  \left(\int_{\left(T^d\right)^2} \frac{1}{N}  \Big\| \Delta_y \sum_{\ell,m=1}^\infty  \phi_\ell(y)\phi_m(y^\prime)a_\ell  a_m  e^{-i H^\aux_N t} \Phi^{\leq N}\Big\|^2 dydy^\prime\right)^{1/2}. \label{eq:quartic-t2} 
\end{align}
Using~\eqref{eq:phi-aa},~\eqref{eq:aaaa}, we have
\begin{equation*}
 \eqref{eq:quartic-t1} \leq \|\Delta^{-1}V\|_{L^2\to L^2} N^{-1/2} \|\cN_+ U_N e^{-i H_N t}\Psi_N\| \leq C N^{-1/2}\alpha(t)^{1/2},
\end{equation*}
since $\Delta^{-1}V$ is a bounded operator by hypothesis.

%
For~\eqref{eq:quartic-t2} we use that $\|\cN_+ e^{-i H^\aux_N t} \Phi^{\leq N}\| \leq N  \| e^{-i H^\aux_N t} \Phi^{\leq N}\|$ to obtain, similarly to~\eqref{eq:phi-aa},
\begin{align}
  \eqref{eq:quartic-t2} &\leq 
 \left(\int_{T^d}  \sum_{\ell,m=1}^\infty  \left\langle  \Delta_y \sum_{\ell=1}^\infty \phi_\ell(y) a_\ell e^{-i H^\aux_N t} \Phi^{\leq N} ,   \Delta_y \sum_{m=1}^\infty \phi_m(y) a_m   e^{-i H^\aux_N t} \Phi^{\leq N}\right\rangle dy\right)^{1/2} \notag \\
  %
  &= \left\langle e^{-i H^\aux_N t} \Phi^{\leq N}, d\Gamma\big(\Delta^2\big) e^{-i H^\aux_N t} \Phi^{\leq N} \right\rangle^{1/2}\notag \\
  &\leq \|d\Gamma(-\Delta)e^{-i H^\aux_N t} \Phi^{\leq N}\|\;.\label{eq:quartic bound}
\end{align}
Now 
\begin{equation*}
 \|d\Gamma(-\Delta)e^{-i H^\aux_N t} \Phi^{\leq N}\| 
 \leq \|H_0 e^{-i H^\aux_N t} \Phi^{\leq N}\|
 \leq \|e^{-i H^\aux_N t} H^\aux_N\Phi^{\leq N}\| 
 + \|(H^\aux_N-H_0)e^{-i H^\aux_N t} \Phi^{\leq N}\|,
\end{equation*}
and the difference $H^\aux_N-H_0$ is a quadratic operator that is $\cN_+$-bounded uniformly in $N$ by Lemma~\ref{lem:quadratic} and~\eqref{eq:V ell2}. 
We thus have the following bound on the quartic term~\eqref{eq:H diff quartic}
\begin{align*}
 \frac{2}{N}\bigg|  \sum_{j,k,\ell,m=1}^\infty \bra U_N e^{-i H_N t} \Psi_N,V_{jk\ell m} a^*_j a^*_k a_\ell a_me^{-i H^\aux_N t} \Phi^{\leq N}\ket\bigg| \leq C N^{-1/2} \sqrt{\alpha(t)}\Big(\| H^\aux_N\Phi^{\leq N}\|+\sqrt{\alpha^\aux(t)}  \Big).
\end{align*}

By a similar argument (see also~\eqref{eq:dGamma bound}), we have the bound
\begin{align*}
 &\frac{1}{\sqrt N}\left|\left\langle U_N e^{-i H_N t} \Psi_N, d\Gamma (QW_x Q) e^{-i H^\aux_N t} \Phi^{\leq N} \right\rangle\right| \\
 &\leq \sup_{x\in T^d} \| \Delta^{-1} W_x \|_{L^2\to L^2} N^{-1/2} \|\cN_+^{1/2} U_N e^{-i H_N t} \Psi_N\| \|d\Gamma(-\Delta)e^{-i H^\aux_N t} \Phi^{\leq N}\| \\
 &\leq C N^{-1/2}\sqrt{\alpha(t)}\Big(\| H^\aux_N\Phi^{\leq N}\|+\sqrt{\alpha^\aux(t)}  \Big).
\end{align*}

Using the assumption that $\Phi\in D(H_0)\subset D(H^\aux)$, Lemma~\ref{lem:alpha} and integrating in~\eqref{eq:diffaux} yields the claim. \eed
\end{proof}

To complete the proof of Theorem~\ref{thm:quant}, it remains to prove the following Lemma on the approximation of $e^{-iH^\aux _Nt}$ by $e^{-i H^\BF t}$, which essentially amounts to removing the restriction to $\cH_+^{\leq N}$.

\begin{lemma}\label{lem:aux-BF}
  Assume the hypothesis of Theorem~\ref{thm:quant} and let $\Phi\in D(H_0)$. There exists a constant $K$ such that for all $N\in \N$ and $t\in \R$
	\begin{equation*}
	\big \| e^{-i H_N^\aux t} \Phi^{\leq N} - e^{-i H^\BF t} \Phi \big\|\leq K e^{v|t|} N^{-1/4},
	\end{equation*}
	where $v$ is the constant of Lemma~\ref{lem:alpha}.
\end{lemma}
\begin{proof}
We have $\Phi \in D(H_0)\subset D(H^\BF)$ (see Proposition~\ref{prop:dom BF}), so $\Phi^{\leq N} \in D(H^\aux_N)$ by~Lemma~\ref{lem:dom aux}. 
 Note that 
\begin{equation}
 \|\Phi-\Phi^{\leq N}\|^2 = \sum_{n=N+1}^\infty \|\Phi^{(n)}\|^2 \leq (N+1)^{-2} \|\cN_+ \Phi\|^2,
\end{equation}
so, by Lemma~\ref{lem:duhamel}, we have
\begin{equation}
 \big \| e^{-i H_N^\aux t} \Phi^{\leq N} - e^{-i H^\BF t} \Phi \big\|^2 \leq (N+1)^{-2} \|\cN_+ \Phi\|^2
 + \int\limits_0^t \Big|\Big\langle e^{-i H_N^\aux s} \Phi^{\leq N}, (H_N^\aux - H^\BF)e^{-i H^\BF s} \Phi \Big\rangle\Big| ds.
\end{equation}
The difference of the operators can be written as
\begin{align}
 H_N^\aux - H^\BF
 =&  -\frac{2}{N}\sum_{j, k>1} V_{j0k0} a_j^*\cN_+a_k  \notag\\
&+ \frac{1}{N}\sum_{j, k>1} V_{jk00} a_j^*(\sqrt{N-\cN_+}-\sqrt{N}) a_k^*( \sqrt{N-\cN_+}-\sqrt N) \notag\\
&+\frac{1}{N}\sum_{j, k>1} V_{00jk}  (\sqrt{N-\cN_+}-\sqrt N)
a_k(\sqrt{N-\cN_+}-\sqrt N) a_k\notag \\
& + \frac{1}{\sqrt{N}} a^*(W_x)(\sqrt{N-\cN_+}-\sqrt N) +  \frac{1}{\sqrt{N}} (\sqrt{N-\cN_+}-\sqrt N)a(W_x)\notag.
\end{align}
Using that 
\begin{align}
 a_k^* \cN_+  &=  (\cN_+ -1)a_k^* \\
  a_k^*\sqrt{N-\cN_+}&=\sqrt{N+1-\cN_+} a_k^* \\
  a_k\sqrt{N-\cN_+} &= \sqrt{N-\cN_+-1} a_k
\end{align}
we can move all of the $\cN_+$-dependent factors to the left. The factors on the right are then terms that also appear in $H^\BF$ and are $\cN_+$-bounded by Lemma~\ref{lem:quadratic} and~\eqref{eq:V ell2}.
With
\begin{align}
 N^{-1}& \|(\sqrt{N+1-\cN_+}\sqrt{N+2-\cN_+} - N) \Psi\|^2 \notag\\
 & \leq  N^{-1}\big\langle ((N+2)^2 +N^2)  \Psi, \Psi \big \rangle 
 - 2  N^{-1}\Re \big\langle N(N+1-\cN_+)  \Psi, \Psi \big \rangle \notag\\
 &\leq 2 N^{-1} \big\langle (2N+16) \Psi, \Psi \big \rangle  + 2 \big\langle \Psi, \cN_+ \Psi \big\rangle \leq C (\|\Psi\| + \|\cN_+ \Psi\|)^2,
\end{align}
we then obtain
\begin{align}
 \Big|\Big\langle e^{-i H_N^\aux s} \Phi^{\leq N}, (H_N^\aux - H^\BF)e^{-i H^\BF s} \Phi \Big\rangle\Big| 
 \leq C N^{-1/2} \big(\alpha^\aux(s)\alpha^\BF(s)\big)^{1/2}
 \end{align}
with some constant $C>0$. Applying Lemma~\ref{lem:alpha} thus completes the proof.
\end{proof}

\appendix
\section{Self-adjointness of the Hamiltonians}

Here we prove the relevant self-adjointness and domain properties of $H^\aux$ and $H^\BF$, as well as a useful general Lemma.

\begin{lemma}\label{lem:quadratic}
 Let $(M_{jk})_{j,k\in \N}\in \ell^2(\N\times \N)$. Then for any $\Phi\in D(\cN_+)$
 \begin{align*}
  \Big\| \sum_{j,k\in \N} M_{j,k} a_j a_k  \Phi\Big\|_{\cF_+}& \leq \|M\|_{\ell^2} \|\cN_+ \Phi\|_{\cF_+} \\
   \Big\| \sum_{j,k\in \N} M_{j,k} a_j^* a_k^*  \Phi\Big\|_{\cF_+}& \leq \|M\|_{\ell^2} \|(\cN_+ +2) \Phi\|_{\cF_+}.
 \end{align*}
\end{lemma}
\begin{proof}
 We only prove the first inequality, the second can be proved in a similar way.
 We have for any $\Psi\in \cF_+$ by the Cauchy-Schwarz inequality
 \begin{equation}
  \Big|\Big\langle \Psi, \sum_{j,k\in \N} M_{j,k} a_j a_k  \Phi \Big\rangle\Big| 
  \leq \|M\|_{\ell^2} \left(\sum_{j,k\in \N} |\langle \Psi, a_j a_k  \Phi \rangle|^2\right)^{1/2}.
 \end{equation}
With \eqref{eq:aaaa} we get the claim.
\end{proof}

\begin{lemma}\label{lem:dom aux}
Assume that $V,W\in L^2(T^d)$ are infinitesimally bounded relative to $-\Delta$. The operator $H_N^\aux$ defined by the expression~\eqref{eq:haux} is self-adjoint on $U_N D(H_N)$ and essentially self-adjoint on $D(H_0)$.
\end{lemma}
\begin{proof}
For the second claim, we prove that $H_N^\aux$ is a perturbation of $-\tfrac{1}{2m}\Delta_x + d\Gamma(-\Delta)^{\leq N}$, the projection of $H_0$ to $\cH_+^{\leq N}$, by a bounded operator.

For the quadratic terms in~\eqref{eq:haux}, this follows from the fact that $H_N^\aux$ acts non-trivially only on $\cH_+^{\leq N}$ and Lemma~\ref{lem:quadratic} together with Parseval's identity, which yields
\begin{align}\label{eq:V ell2}
 \|V_{jk00}\|_{\ell^2(\N\times \N)}^2&= \|V_{00jk}\|_{\ell^2(\N\times \N)}^2 \notag\\
 &  \leq \sum_{j,k\in \N_0} \Big| \int_{T^d } \int_{T^d }  \bar\phi_j(y) \bar\phi_k(y') V(y-y') dy  dy' \Big|^2 \notag\\
 &= \|V(y-y')\|_{L^2(T^d \times T^d)}^2 = \|V\|_{L^2(T^d) }^2, \notag\\
 \|V_{j0k0}\|_{\ell^2(\N\times \N)}^2&\leq \sum_{j,k\in \N_0} \big|\langle \phi_j \otimes \phi_0, V(y_1-y_2) \phi_0 \otimes \phi_k \rangle\big|^2 =  \|V\|_{L^2(T^d) }^2.
\end{align}
For the linear term in~\eqref{eq:haux} this follows from the bound $\|a(W_x)\Psi\|_{\cH_+}\leq \|W\|_{L^2} \|\sqrt{\cN_+} \Psi \|_{\cH_+}$  by the same reasoning. Hence, $H^\aux_N$ is self-adjoint on the domain of $-\tfrac{1}{2m}\Delta_x + d\Gamma(-\Delta)^{\leq N}$ and essentially self-adjoint on $D(H_0)$ by the Kato-Rellich theorem.  

To obtain the first claim it is now sufficient to prove that the difference of $H_N^\aux$ and $U_N H_N U_N^*$ is bounded relative to $d\Gamma(-\Delta)^{\leq N}$, with relative bound zero.
This difference consists of cubic~\eqref{eq:H diff cubic} and quartic~\eqref{eq:H diff quartic} terms involving $V$, and a quadratic term~\eqref{eq:H diffW} with $W$. 
The cubic terms are bounded by an argument analogous to~\eqref{eq:V cubic bound}. 
The relative bound for $N^{-1/2}d\Gamma(QW_xQ)^{\leq N}$ (i.e.~\eqref{eq:H diffW}) is a consequence of the bound for $\Psi \in D(d\Gamma(-\Delta))$, $\Phi \in \cH_+^{\leq N}$, and $\lambda\geq0$ (cf.~\cite[Eqs.2.9, 2.10]{moller2005})
\begin{align}
 |\langle  \Phi, N^{-1/2} d\Gamma(QW_xQ)^{\leq N} \Psi\rangle| 
 &\leq N^{-1/2}\| \cN_+^{1/2} \Phi\|_{\cH_+^{\leq N}}   \|d\Gamma((QW_xQ)^2)^{1/2}\Psi \|_{\cH_+^{\leq N}} \notag \\
 &\leq \eps \|\Phi\|_{\cH_+^{\leq N}} \|d\Gamma((\lambda-\Delta)^2)^{1/2}\Psi \|_{\cH_+^{\leq N}}  \notag \\
 &\leq  \eps \|\Phi\|_{\cH_+^{\leq N}}\|d\Gamma(\lambda-\Delta)\Psi \|_{\cH_+^{\leq N}}, \label{eq:dGamma bound}
\end{align}
with $\eps=\|(\lambda-\Delta)^{-1}W_x\|_{L^2 \to L^2}=\|W_x(\lambda-\Delta)^{-1}\|_{L^2 \to L^2}$.
Since $\eps$ goes to zero as $\lambda\to \infty$, because $W$ is infinitesimally $-\Delta$-bounded, this gives an infinitesimal bound.
The quartic term is $d\Gamma(-\Delta)^{\leq N}$-bounded by the argument that gives~\eqref{eq:quartic bound}, but replacing $-\Delta$ by $\lambda-\Delta$ and arguing as above.
\end{proof}

\begin{proposition}\label{prop:dom BF}
 Let $V, W\in L^2(T^d)$, then $H^\BF-H_0$ is $\cN_+$-bounded and $H^\BF $ is essentially self-adjoint on $D(H_0)$. 
\end{proposition}
\begin{proof}
The first statement follows from Lemma~\ref{lem:quadratic} as above.
Consequently, there exists a constant $c$ such that $H^\BF+ c\cN_+$ is self-adjoint and positive on $D(H_0)=D(H_0+c\cN_+)$ by the Kato-Rellich theorem, since $H^\BF-H_0$ is $H_0 + c\cN_+$-bounded with bound less than one.

Essential self-adjointness of $H^\BF$ can now be obtained by applying the commutator theorem~\cite[Thm.X.36]{ReSi2}, with  $H^\BF + c\cN_+$ as a comparison operator. For this, it is sufficient to prove that
\begin{equation*}
 |\langle \Phi, [\cN_+, H^\BF ] \Phi\rangle | = |\langle \Phi, [\cN_+, H^\BF-H_0 ] \Phi\rangle | \leq C \langle \Phi, \cN_+ \Phi \rangle,
\end{equation*}
for some constant $C>0$ and all $\Phi \in D(H_0^{1/2})$.
Since the commutator of $\cN_+$ with $H^\Bog$ is again a quadratic operator, composed of the same terms up to signs, this follows from Lemma~\ref{lem:quadratic} as above.
\end{proof}


\end{document}